\documentclass[a4paper,UKenglish,cleveref, autoref]{lipics-v2019}
\newcommand{\h}{\hspace*{0.2in}}
\newcommand{\quotes}[1]{``$#1$''}
\nolinenumbers
\usepackage{thmtools,thm-restate}
\usepackage{tikz}
\usepackage{makecell}
\usetikzlibrary{positioning}
\usetikzlibrary{decorations.pathreplacing}
\title{Byzantine Lattice Agreement in Asynchronous Systems} 

\titlerunning{Byzantine Lattice Agreement in Asynchronous Systems}

\author{Xiong Zheng}{University of Texas at Austin, Austin, TX 78712, USA}{zhengxiongtym@utexas.edu}{}{}

\author{Vijay K. Garg}{University of Texas at Austin, Austin, TX 78712, USA}{garg@ece.utexas.edu}{}{}

\authorrunning{X.\,Zheng, and V.\,K. Garg} 

\Copyright{Xiong Zheng, and Vijay K. Garg}



\category{}

\relatedversion{}

\supplement{}

\acknowledgements{}

\begin{document}

\maketitle

\begin{abstract}
We study the Byzantine lattice agreement (BLA) problem in asynchronous distributed message passing systems. In the BLA problem, each process proposes a value from a join semi-lattice and needs to output a value also in the lattice such that all output values of correct processes lie on a chain despite the presence of Byzantine processes. We present an algorithm for this problem with round complexity of $O(\log f)$ which tolerates $f < \frac{n}{5}$ Byzantine failures in the asynchronous setting without digital signatures, where $n$ is the number of processes. We also show how this algorithm can be modified to work in the authenticated setting (i.e., with digital signatures) to tolerate $f < \frac{n}{3}$ Byzantine failures. 
\end{abstract}

\section{Introduction}
In distributed systems, reaching agreement in the presence of process failures is a fundamental task. Understanding the kind of agreement that can be reached helps us understand the limitation of distributed systems with failures. Consensus~\cite{pease1980reaching} is the most fundamental problem in distributed computing. In this problem, each process proposes some input value and has to decide on some output value such that all correct processes decide on the same valid output. In synchronous message systems with crash failures, consensus cannot be solved in fewer than $f + 1$ rounds~\cite{dolev1983authenticated}. In asynchronous systems, consensus is impossible in the presence of even one crash failure~\cite{fischer1985impossibility}. The $k$-set agreement~\cite{chaudhuri1993more} is a generalization of consensus, in which processes can decide on at most $k$ values instead of just one single value. The $k$-set agreement cannot be solved in asynchronous systems if the number of crash failures $f \geq k$~\cite{biely2012s, herlihy1998unifying}. The paper~\cite{chaudhuri2000tight} shows that $k$-set agreement problem cannot be solved within $\lfloor \frac{f}{k} \rfloor$ rounds if $n \geq f + k + 1$ in crash failure model. 

The lattice agreement problem was proposed by Attiya et al~\cite{attiya1995atomic} to solve the atomic snapshot object problem in shared memory systems. In this problem, each process $i \in [n]$ has input $x_i$ and needs to output $y_i$ such that the following properties are satisfied. 

1) {\bf Downward-Vaility}: $x_i \leq y_i$ for each correct process $i$. 

2) {\bf Upward-Validity}: $y_i \leq \sqcup \{x_i ~ | ~ i \in [n]\}$. 

3) {\bf Comparability}: for any two correct processes $i$ and $j$, either $y_i \leq y_j$ or $y_j \leq y_i$. \\

The lattice agreement problem is a weaker problem than the consensus problem and the $k$-set agreement problem. It can be solved in $O(\log f)$ rounds in synchronous systems and tolerate $f < n$ crash failures \cite{zheng2018lattice}. In asynchronous systems, it can also be solved in $O(\log f)$ rounds but can only tolerate $f < \frac{n}{2}$ crash failures \cite{zheng2018linearizable}.

Attiya et al in~\cite{attiya1995atomic} presents a generic algorithm to transform any protocol for the lattice agreement problem to a protocol for implementing an atomic snapshot object in shared memory systems. This transformation can be easily implemented in message passing systems by replacing each read and write step with sending \quotes{read} and \quotes{write} messages to all and waiting for acknowledgements from $n - f$ different processes. On the other hand, if we can implement an atomic snapshot object, lattice agreement can also be solved easily both on shared-memory and message passing systems with only crash failures. Thus, solving the lattice agreement problem in message passing systems is equivalent to implementing an atomic snapshot object in message passing systems with only crash failures. 

Using lattice agreement protocols, Faleiro et al~\cite{faleiro2012generalized} gives procedures to build a special class of linearizable and serializable replicated state machines which only support query operations and update operations but not mixed query-update operations. Later, Xiong et al~\cite{zheng2018linearizable} proposes some optimizations for their procedure for implementing replicated state machines from lattice agreement in practice. They propose a method to truncate the logs maintained in the procedure in~\cite{faleiro2012generalized}. 
The recent paper \cite{skrzypczak2019linearizable} by Skrzypczak et al proposes a protocol based on generalized lattice agreement, which is a multi-shot version of lattice agreement problem, to provide  linearizability for state based conflict-free data types~\cite{shapiro2011conflict}  the procedure given in \cite{faleiro2012generalized} in terms of memory consumption, at the expense of progress. 

In message passing systems with crash failures, the lattice agreement problem is well studied ~\cite{attiya1995atomic, zheng2018linearizable, zheng2018lattice, mavronicolasabound}. The best upper bound for both synchronous systems and asynchronous systems is $O(\log f)$ rounds. In the Byzantine failure model, a variant of the lattice agreement problem is first studied by Nowak et al~\cite{nowak2019byzantine}. Then, Di Luna et al~\cite{di2019byzantine} proposes a validity condition which still permits the application of lattice agreement protocol in obtaining atomic snapshots and implementing a special class of replicated state machines. They present an $O(f)$ rounds algorithm for the Byzantine lattice agreement problem in asynchronous message systems. For synchronous message systems, a recent preprint~\cite{zheng2019byzantine} gives three algorithms. The first algorithm takes $O(\sqrt{f})$ rounds and has the early stopping property. The second and third algorithm takes $O(\log n)$ and $O(\log f)$ rounds but are not early stopping. All three algorithms can tolerate $f < \frac{n}{3}$ failures. They also show how to modify their algorithms to work for authenticated settings and tolerates $f < \frac{n}{2}$ failures. The preprint \cite{di2020synchronous} presents an algorithm which takes $O(\log f)$ rounds which can tolerate $f < \frac{n}{4}$ failures and shows how to improve resilience to $f < \frac{n}{3}$ by using digital signatures and public-key infrastructure.  
 
In this work, we present new algorithms for the Byzantine lattice agreement (BLA) problem in asynchronous message passing systems. In this problem, each process $i \in [n]$ has input $x_i$ from a join semi-lattice $(X, \leq, \sqcup)$ with $X$ being the set of elements in the lattice, $\leq$ being the partial order defined on $X$, and $\sqcup$ being the join operation. Each process $i$ has to output some $y_i \in X$ such that the following properties are satisfied. Let $C$ denote the set of correct processes in the system and $t$ denote the actual number of Byzantine processes in the system. 

\textbf{Comparability}: For all $i \in C$ and $j 
\in C$, either $y_i \leq y_j$ or $y_j \leq y_i$.

\textbf{Downward-Validity}: For all $i \in C$, $x_i \leq y_i$. 

\textbf{Upward-Validity}: $\sqcup \{y_i ~ | ~ i \in C\} \leq \sqcup(\{x_i ~ | ~ i \in C\} \cup B)$, where $B \subset X $ and $|B| \leq t$.  \\

Our main contribution is summarized as follows. 

\begin{restatable}{theorem}{logTheoremAsync} \label{log_theorem_async}
There is an $O(\log f)$ rounds algorithm for the BLA problem in asynchronous systems which can tolerate $f < \frac{n}{5}$ Byzantine failures, where $n$ is the number of processes in the system. The algorithm takes $O(n^2 \log f)$ messages. 
\end{restatable}

\begin{restatable}{theorem}{logTheoremAsyncAuth} \label{log_theorem_async_auth}
There is a $O(\log f)$ rounds algorithm for the BLA problem in authenticated asynchronous systems which can tolerate $f < \frac{n}{3}$ Byzantine failures, where $n$ is the number of processes in the system. The algorithm takes $O(n^2 \log f)$ messages. 
\end{restatable}



\section{System Model}
We assume a distributed asynchronous message system with $n$ processes with unique ids in $[1,2,...,n]$. The communication graph is a clique, i.e., each process can send messages to any other process in the system. We assume that the communication channel between any two processes is reliable. There is no upper bound on message delay. We assume that processes can have Byzantine failures but at most $f < n/3$ processes can be Byzantine in any execution of the algorithm. We use parameter $t$ to denote the actual number of Byzantine processes in a system. By our assumption, we must have $t \leq f$. Byzantine processes can deviate arbitrarily from the algorithm. We say a process is correct or non-faulty if it is not a Byzantine process. We consider both systems with and without digital signatures. In a system with digital signatures, Byzantine processes cannot forge the signature of correct processes. 

\section{$O(\log f)$ Rounds Algorithm for the Asynchronous BLA Problem}
In this section, we present an algorithm for the BLA problem in asynchronous systems which takes $O(\log f)$ rounds of asynchronous communication and tolerates $f < \frac{n}{5}$ Byzantine failures. The high level idea of the algorithm is to apply a Byzantine tolerant classifier procedure, similar to the classifier procedure in \cite{zheng2018lattice} for crash failure systems, to divide a group of processes into the slave subgroup and the master subgroup such that the values of the slave group is less than the values of the master group. Then, by recursively applying such a classifier procedure within each subgroup, eventually all processes have comparable values. The classifier procedure for crash failure model only needs to guarantee that the value a correct slave process is at most the value of any correct master process and the size of the union of all values of correct slave processes is at most $k$. The parameter $k$ is a knowledge threshold. 

Fig. \ref{fig:clsTree} shows the classification tree which specifies the threshold parameter for each group when the classifier procedure is invoked recursively. Before recursively invoking the classifier procedure, an initial round is used to let all processes exchange their input values. After this initial round, each process obtains at least $n - f$ values. Each node in the tree represents a group. We also use label to indicate the threshold parameter of a group. Initially, all processes are in the same group with label $n - \frac{f}{2}$. Let $G$ be a group with label $k$ at level $r$. The master group (the right child node) of $G$ has label $k + \frac{f}{2^{r + 1}}$. The slave group (the left child node) of $G$ has label $k - \frac{f}{2^{r + 1}}$. We can observe that all labels in the classification tree up to level $\log f$ are unique. The above properties of the classifier procedure guarantee that processes in a leaf group must have the same value. 

\begin{figure}[htb]
\begin{minipage} {0.5\textwidth}
\begin{tikzpicture}[
roundnode/.style={circle, draw=green!60, very thick, minimum size=7mm},
covernode/.style={circle, draw=green!60, fill=green!60, very thick, minimum size=7mm},
]

\node[roundnode,label=above:{$n - \frac{f}{2}$}]			(G)       					    {};
\node[roundnode,label=above:{$n - \frac{3f}{4}$}]			(S)        [below left=0.4cm and 0.8 of G] 		
		{};
\node[roundnode,label=above:{$n - \frac{f}{4}$}]			(M)			[below right= 0.4cm and 0.8 of G]		
		{};
\node[roundnode,label=below:{$n-f$}]			
(l1)			[below left= 0.4 and 0.8 of S]		
		{};
\node[roundnode,label=below:{$n-f + 1$}]			
(l2)			[right = 0.6 of l1]		
		{};

\node[roundnode,label=below:{$n$}]			
(lf)			[below right= 0.4 and 0.8 of M]		
		{};
		\node[roundnode,label=below:{$n-1$}]			
(lf1)			[left = 0.6 of lf]		
		{};
\node[left = 2.8 of G] (r1) {$level~ 1:$};
\node[left = 1.5 of S] (r2) {$level ~2:$};
\node[left = 0.2 of l1] (r3) {$level ~ \log f + 1:$};
\path (r2) -- node[auto=false,thick]{\vdots} (r3);
\draw[-] (G) -- (S);
\draw[-] (G) -- (M);
\draw[dashed] (S) -- (l1);
\draw[dashed] (M) -- (lf);
\hspace{0.025in} 
\path (l2) -- node[auto=false,thick]{\ldots} (lf1);

\end{tikzpicture}
\end{minipage}
\caption{The Classification Tree}
\label{fig:clsTree}
\end{figure}
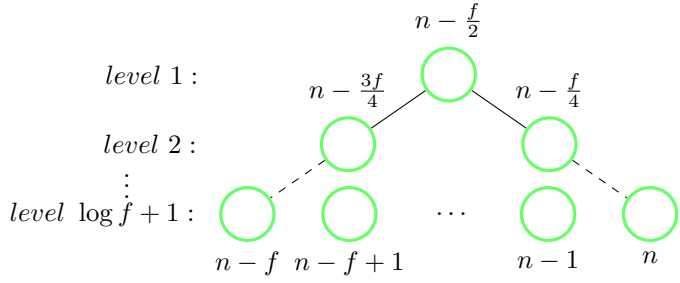

In presence of Byzantine processes, the above properties are not enough for recursively applying such classifier procedure within each subgroup. A Byzantine process in a slave group can introduce some new values which are not known by some master process. We introduce the notion of admissible values for a group (to be defined later), which is the set of values that processes in this group can ever have. We present a Byzantine tolerant classifier procedure with threshold parameter $k$ which provides the following three properties: 1) Each correct slave process has at most $k$ values and each correct master process has more than $k$ values. 2) The admissible values of the slave group is a subset of the value of any correct master process. 3) The union of all admissible values in the slave group has size less than the threshold parameter $k$.  

Suppose now we have a classifier which guarantees the above properties. The main algorithm, shown in Fig. \ref{fig:main_algorithm}, proceeds in asynchronous rounds.  Each process $i$ has a label $l_i$, which is used as the threshold parameter when it invokes the classifier procedure. Initially, each process has the same label $k_0 = n - \frac{f}{2}$. The label of a process is updated at each round according to the classification tree. In some places of our algorithm, we use the reliable broadcast primitive defined by Bracha~\cite{bracha1987asynchronous} to send values. In this primitive, a process uses {\bf RB\_broadcast} to send a message and uses {\bf RB\_broadcast} to reliably deliver a message. This primitive guarantees many nice properties. In our algorithm, we need the following two main properties: 1) If a message is reliably delivered by some correct process, then this message will eventually be reliably delivered by each correct process. 2) If a correct process reliably delivers a message from process $p$, then each correct process reliably delivers the same message from $p$. 

In the initial round at line 1-2, process $i$ {\bf RB\_broadcast} its input $x_i$ to all and waits for {\bf RB\_deliver} from $n-f$ different processes. Then, it updates its value set to be the set of values reliably delivered at this round. When reliable delivering a value, process $i$ adds this value into its safe value set for the initial group $k_0 = n - \frac{f}{2}$. The reliable delivering procedure is running on background. So this safe value set for the initial group keeps growing. By the properties of reliable broadcast, this safe value set can only contain at most one value from each process. This ensures {\bf Upward-Validity}. 

After the initial round, we can assume that all values in the initial safe value set of each process are unique, which can be done by associating the sender's id with the value. 

From line 3 to line 8, process $i$ executes the classifier procedure (to be presented later) for $\log f$ rounds. At each round, it invokes the classifier procedure to decide whether it is classified as a slave or master and then updates its value accordingly. At round $r$, if process $i$ is a master, it updates its label to be $l_i := l_i + \frac{f}{2^{r + 1}}$. Otherwise, if updates its label to be $l_i := l_i - \frac{f}{2^{r + 1}}$. 

Applying the above three properties provided by the classifier procedure, by induction on the round number, we can readily see that at the end of round $\log f$, for any two correct process $i$ and $j$, if they are in the same group, say with label $k$, then both $i$ and $j$ must have exactly $k$ values and their union also has exactly $k$ values. Then, $i$ and $j$ must have the same set of values. If they are in different group, then by recursively applying property 2, the values of one process must be subset of the values of the other process. 

We now show a Byzantine tolerant classifier which satisfies the properties we define. 
\begin{figure} [htbp] 
\fbox{\begin{minipage}[t]  {4in}
\noindent
\underline{Code for process $i$:} \\
$x_i$: input value \h $y_i$: output value \\
$l_i$: label of process $i$. Initially, $l_i = k_0 = n - \frac{f}{2}$:  \\
$V_i^r$: value set held by process $i$ at round $r$ of the algorithm\\ 
Map $S_i$: $S_i[k]$ denote the safe value set for group $k$\\

 /* Initial Round */ \\
{\bf 1:} {\bf RB\_broadcast}($x_i$), wait for $n - f$ {\bf RB\_deliver}($x_j$) from $p_j$  \\
{\bf 2:} Set $V_i^1$ as the set of values reliably delivered \\

/* Round 1 to $\log f$ */ \\
{\bf 3:} {\bf for} {$r := 1$ to $\log f$}\\
{\bf 4:} \h ($V_i^{r + 1}, class$) := $Classifier(V_i^{r}, l_i, r)$\\
{\bf 5:} \h {\bf if} {$class = master$} \h {\bf then} $l_i := l_i + \frac{f}{2^{r + 1}}$\\
{\bf 6:} \h {\bf else} \h $l_i := l_i - \frac{f}{2^{r + 1}}$\\
{\bf 7:} {\bf end for}\\
{\bf 8:} $y_i := \sqcup \{v \in V_i^{\log f + 1}\}$ \\
 
 {\bf Upon RB\_deliver}$(x_j)$ from $p_j$ \\
 \h $S_i[k_0] := S_i[k_0] \cup x_j$
\end{minipage}
}
\caption{$O(\log f)$ Rounds Algorithm for the BLA Problem} \label{fig:main_algorithm}
\end{figure}

\subsection{The Byzantine Tolerant Classifier}
The Byzantine tolerant classifier procedure, shown in Fig. \ref{fig:classifier}, is inspired by the asynchronous classifier procedure given in \cite{zheng2018linearizable} for the crash failure model. We say a process writes a value to at least $n - f$ processes if it sends a \quotes{write} message containing the value to all processes and waits for $n - f$ different processes to send acknowledgement back. We say a process reads from at least $n - f$ processes if it sends a \quotes{read} message to all processes and waits for at least $n - f$ processes to send their current values back. We say a process performs a write-read step if it writes its value to at least $n - f$ processes and reads their values. 

In the asynchronous classifier procedure for the crash failure model \cite{zheng2018linearizable}, to divide a group into a slave subgroup and a master subgroup, each process in the group first writes its value to at least $n - f$ processes and then reads from at least $n - f$ processes. After that, each process checks whether the union of all values read has size greater than the threshold parameter $k$ or not. If true, it is classified as a master process, otherwise, it is classified as a slave process. Slave processes keep their value the same. To guarantee that the value of each slave process is less than the value of each master process, each master process performs a write-read step to write the values obtained at the read step to at least $n - f$ processes and read the values from them. Then it updates its value to be the union of all values read. The second read step guarantees that the size of the union of values of slave processes is at most $k$, since the last slave process which completes the write step must be able to read all values of slave processes. 

Constructing such a classifier procedure in presence of Byzantine processes is much more difficult. In order to adapt the above procedure to work in Byzantine failure setting, we need to address the following challenges. First, in the write step or read step, when a process waits for at least $n - f$ different processes to send their values back, a Byzantine process can send arbitrary values. Second, simply ensuring that the values of a slave process is a subset of values of each master process is not enough, since a Byzantine process can introduce some values unknown to a master process in the slave group. For example, even if we can guarantee that the current value of each slave process is less that the value of each master process, in a later round, a Byzantine process can send some new value to a slave process which is unknown to some master process. This is possible in an asynchronous systems since messages can be arbitrarily delayed. Third, ensuring that the union of all values in the slave group has size at most $k$ is quite challenging. A simple second read step does not work any more since the last process which completes the write step might be a Byzantine process. 

To prevent the first problem, in the Byzantine classifier procedure, when a process wants to perform a write step or read step, it applies the reliable broadcast primitive to broadcast its value. When a process waits for values from at least $n - f$ processes, it only accepts a value if the value is a subset of the values reliably delivered by this process. By property of reliable broadcast, this ensures that each accepted value must be reliably broadcast by some process, which prevents Byzantine processes from introducing arbitrary values. 

To tackle the second and third problem, the key idea is to restrict the values that a Byzantine process, which claims itself to be a slave process, can successfully reliable broadcast in later rounds. To achieve that, first we require that that a slave process can only reliable broadcast the value that it has reliably broadcast in the previous round. This prevents Byzantine processes from introducing arbitrary new values into a slave group. Second, we require each process which claims itself as a slave process to prove that it is indeed classified as a slave at the previous round when it tries to reliable broadcast a value at the current round by presenting the set of values it used to do classification. To enforce the above two requirements, we add a validity condition when a process echoes a message in the reliable broadcast primitive.  However, this is not enough, since the value of a Byzantine slave process might not be known to a master process if broadcast value of the Byzantine process is arbitrarily delayed. To ensure that the value a Byzantine process reliably broadcast to be read by each correct master process, we force a Byzantine process who wants to be able to reliable broadcast a value in the slave group at next round to actually write its value to at least $\lfloor \frac{n + f}{2} \rfloor + 1 - f$ correct processes, i.e., at least $\lfloor \frac{n + f}{2} \rfloor + 1 - f$ correct processes must have received the value of a Byzantine process before each correct master process tries to read from at least $n - 2f$ correct processes. These two sets of correct processes must have at least one correct process in common since $f < \frac{n}{5}$. 

The validity condition we add in the reliable broadcast ensures that the admissible values for the slave group must be subset of the value set of any correct master process and the union of these values has size at most the threshold $k$. 

Each process which is classified as master is not required to prove its group identity but the value it tries to broadcast has to be a subset of safe value sets of correct processes. Similar to the asynchronous classifier procedure in \cite{zheng2018linearizable}, to ensure that the value of a slave process is less than the value of a master process, the master process needs to write to and read from at least $n - f$ processes after it is classified as a master process.

\subsubsection{Bounded Reliable Broadcast}
 \begin{figure}[htb] 
\fbox{\begin{minipage}[t]  {5.6in}
\noindent
\underline{{\bf BRB\_broadcast($type, pf, v, k, r$)}} \\
$type$ denotes the type of the message to be sent, either \quotes{write} or \quotes{read} \\
$pf$ is an array which is a proof of sender's group identity \\
$v$ is the value to be sent, $k$ is the label of the sender, $r$ is the round number \\
The \textit{valid} function is defined in Fig. \ref{fig:valid_function}\\

{\bf Broadcast} INIT$(i, type, pf, v, k, r)$ to all \\

{\bf Upon} receiving INIT$(j, t_j, pf_j, v_j, k_j, r_j)$\\
\h {\bf if} (first reception of INIT$(j, t_j, -, -,-, r_j)$ \\
\h\h {\bf wait until} $valid(t_j, pf_j, v_j, k_j, r_j)$ \\
\h\h {\bf Broadcast} ECHO$(j, t_j, v_j, k_j, r_j)$ to all \\ 
\h {\bf endif} \\

{\bf Upon} receiving ECHO$(j, t_j, pf_j, v_j, k_j, r_j)$\\
\h {\bf if} ECHO$(j, t_j, pf_j, v_j, k_j, r_j)$ is received from at least $\lfloor \frac{n + f}{2} \rfloor + 1$ different processes \\ \h\h $\wedge$ READY$(j, v_j, k_j, r_j)$ has not yet broadcasted\\
\h\h {\bf Broadcast} READY$(j, t_j, pf_j, v_j, k_j, r_j)$ \\
\h {\bf endif} \\

{\bf Upon} receiving READY$(j, t_j, pf_j, v_j, k_j, r_j)$ \\
\h {\bf if} READY$(j, t_j, pf_j, v_j, k_j, r_j)$ received from $f + 1$ different processes $\wedge$ \h\h \h READY$(j, t_j, pf_j, v_j, k_j, r_j)$ has not been broadcasted\\ 
\h\h {\bf Broadcast} READY$(j, t_j, pf_j, v_j, k_j, r_j)$ \\
\h {\bf endif} \\
\h {\bf if} READY$(j, t_j, pf_j, v_j, k_j, r_j)$ received from $2f+ 1$ different processes $\wedge$ $(j, t_j, pf_j, v_j, k_j, r_j)$ \h\h has not been delivered\\ 
\h\h {\bf BRB\_deliver}$(j, t_j, pf_j, v_j, k_j, r_j)$ \\
\h {\bf endif} 
\end{minipage}
} 

\caption{Bounded Reliable Broadcast \label{fig:brb}}
\end{figure}

Before explaining the Byzantine classifier procedure in detail, we modify the reliable broadcast primitive by adding a condition when a process echoes a broadcast message. This condition (to be explained later) restricts the admissible values for each group. For completeness, the modified reliable broadcast procedure is shown in Fig. \ref{fig:brb}. When a process reliable broadcast a value, it also includes the round number, the label of the group it belongs to and a proof of its group identity. The proof is an array of size $n$ denoting the values read by the sender at previous round. 
When a process $i$ receives a broadcast message from process $j$, it waits for the validity condition to hold and then echoes the message. 
We say a process BRB\_broadcasts a message if it executs {\bf BRB\_broadcast} procedure with the massage. We say a process BRB\_delivers a message if it executes {\bf BRB\_deliver} with this message.  


\subsubsection{Group and Admissible Values}
In our algorithm, each process $i$ has a label $l_i$, which serves as the threshold when it invokes the classifier procedure. The notion of group defined as below is based on labels of processes. 
 
 \begin{definition}[group]
A $group$ is a set of processes which have the same label. The label of a group is the label of the processes in this group. The label of a group is also the threshold value processes in this group use to do classification. 
\end{definition}
We also use label to indicate a group. We say a process is in group $k$ if its message is associated with label $k$. Initially all processes are within the same group with label $k_0 = n - \frac{f}{2}$. The label of each process is updated at each round based on whether it is classified as a slave or a master. 

We introduce the notion of admissible values for a group, which are the set of values that processes in the group can ever have. 
\begin{definition}[admissible values for a group]
The admissible values for a group $G$ with label $k$ is the set of values that can be reliably delivered with label $k$ if they are reliably broadcast by some process (possibly Byzantine) with label $k$.
\end{definition}

In our classifier, each process in group $k$ updates its value set to a subset of the values which are reliably delivered with label $k$. Thus, the value set of each process in group $k$ must be a subset of the admissible values for group $k$. Our algorithm ensures that this property holds continues to hold until the end of the algorithm.


\subsection{The Classifier}
The classifier procedure for process $i \in [n]$, shown in Fig. \ref{fig:classifier}, has three input parameters: $V$ is the current value set of process $i$, $k$ is the threshold value used to do the classification, which is also the current label of process $i$, and $r$ is the round number. 

\begin{figure}[htbp] 
\fbox{\begin{minipage}[t]  {5.6in}
\noindent
\underline{{\bf \textit{Classifier}$(V, k, r)$ for $p_i$:} } \\
$V$: input value set \h $k$: threshold value \h $r$: round number\\
/* Each process $i \in [n]$ keeps track of the following variables */ \\
Array $LB_i^r$. $LB_i^r[j]$ denotes the label of process $j$ sent along its values at round $r$\\
Map $S_i$. $S_i[k]$ denotes a safe value set for group $k$ \\
Map $ACV_i^r$. $ACV_i^r[k]$ denotes the set of values accepted with label $k$, initially $ACV_i^r[k] := \emptyset$ \\
Map $RV_i^r$. $RV_i^r[j]$ denote the values process $i$ read from process $j$ at round $r$.\\
Map $RT_i^r$. $RT_i^r[j]$ denote the values process $j$ read from process $i$ at round $r$.\\

/* write step*/\\
{\bf 1: } {\bf if} $isSlave(k, r)$ {\bf then} $pf := RV_i^{r - 1}$  {\bf else} $pf := \emptyset$ \\
{\bf 2: } BRB\_broadcast$("write", pf, v, k, r)$, wait for $wack(-, r)$ from $n - f$ different processes \\

/* read step*/\\
{\bf 3: } BRB\_broadcast$(``read", -, -, k, r)$, wait for $n - f$ $rack(R_j, r)$ s.t. $R_j \subseteq ACV_i^r[k]$ from $p_j$ \\
{\bf 4: } Set $RV_i^r[j] := R_j$ if $R_j \subseteq ACV_i^r[k]$, otherwise $RV_i^r[j] := \emptyset$\\

/* Classification */\\
{\bf 5: } Let $T_ := \bigcup \limits_{j = 1}^{n} RV_i^r[j]$\\
{\bf 6: } {\bf if} {$|T| > k$} ~/* height is greater than its label */ \\
\h\h ~/* write-read step */ \\
{\bf 7: } \h Send $master(T, k, r)$ to all, wait for $n - f$ $mack(R_j, r)$ from $p_j$ s.t. $R_j \subseteq ACV_i^r[k]$\\
{\bf 8: } \h Define $T' := \cup \{R_j ~| ~ R_j \subseteq ACV_i^r[k], j \in [n]\}$ \\
{\bf 9: } \h {\bf return} ($T'$, \textit{master})\\
{\bf 10:} {\bf else}\\
{\bf 11:} \h{\bf return} ($V$, \textit{slave})\\
\end{minipage}
} 

\fbox{\begin{minipage} {5.6in}
\bf{Upon} BRB\_Deliver$(j, type, -, v, k, r)$\\
\h {\bf if} $type = ``write"$ \\
\h\h $S_i[m(k, r)] := S_i[m(k, r)] \cup v$ /* Construct safe value set for group $m(k, r)$ */\\
\h\h $ACV_i^{r}[k] := ACV_i^{r}[k] \cup v$\\
\h\h $LB_i^{r}[j] := k$ /* Record the label of a process at round $r$ */\\
\h\h Send $wack(-, r)$ to $p_j$\\
\h {\bf elif} $type = ``read"$ \\
\h\h $RT_i^{r}[j] := ACV_i^{r}[k]$ \\
\h\h Send $rack(ACV_i^{r}[k], r)$ to $p_j$ \\
\h {\bf endif} \\

\bf{Upon} receiving $master(T, k, r)$ from $p_j$\\
\h {\bf wait until} $T \subseteq ACV_i^r[k]$ \\
\h Send $mack(ACV_i^r[k], r)$ to $p_j$ 
\end{minipage}
} 
\caption{\textit{The Byzantine Classifier Procedure} \label{fig:classifier}}
\end{figure}

In line 1-2, process $i$ writes its current value set to at least $n - f$ processes by using the {\bf BRB\_broadcast} procedure to send a \quotes{write} message. If process $i$ is classified as a slave at the previous round, it needs to include the array of values it read from at least $n - f$ processes at previous round as a proof of its group identity. This proof is used by every other process in the \textit{valid} function to decide whether to echo the \quotes{write} message or not. When process $i$ BRB\_delivers a \quotes{write} message with label $k$ at round $r$, it includes the value in it into its safe value set for group $m(k,r)$. The safe value set is used to restrict the set of values that can be delivered in the master group $m(k, r)$. Due to this step, we can see that the admissible values in the master subgroup must be a subset of the admissible values at the current group. Process $i$ also includes the value contained in the \quotes{write} message into $ACV_i^r[k]$, which stores the set of values reliably delivered with label $k$ at round $r$. 

From line 3 to line 4, process $i$ reads values from at least $n - f$ processes by using the {\bf BRB\_broadcast} procedure to send a \quotes{read} message to all. In the \textit{valid} function, each process $j$ echos a \quotes{read} message from process $i$ only if it has BRB\_delivered the \quotes{write} message from process $i$ sent at line 2. This step is used to ensure that for any process, possibly Byzantine, to read from other processes, it must have written its value to at least $\lfloor \frac{n + f}{2} \rfloor + 1 - f$ correct processes, otherwise it cannot have enough processes echo its \quotes{read} message in the {\bf BRB\_broadcast}. When process $i$ BRB\_delivers a \quotes{read} message with label $k$ from process $j$ at round $r$, it records the set of values it has reliably delivered with label $k$ in $RT_i^r[j]$. Then process $i$ sends back a $rack$ message along with the set of reliably delivered values with label $k$ at round $r$ to process $j$. At line 3, after the \quotes{read} message is sent, process $i$ has to wait for valid $rack$ message from $n - f$ processes. A $rack$ message is valid if the value set contained in it is a subset of $ACV_i^r[k]$, which is the set of values reliably delivered with label $k$ at round $r$. Consider a $rack(R_j, r)$ message from a correct process $j$. Since $j$ is correct, each value in $R_j$ must have been reliably delivered by process $j$. By property of reliable broadcast, each value in $R_j$ will eventually be reliably delivered by process $i$, thus $R_j \subseteq ACV_i^r[k]$. Thus, eventually process $i$ can obtain $n - f$ valid $rack$ message. At line 4, process $i$ records the set of valid $R_j$s obtained at line 3 into array $RV_i^r$. So, this array stores the values reliably delivered with label $k$ that process $i$ read from all processes. This array is used to do classification in line 5-11 and also used as the proof of group identity of process $i$ when it writes at next round. 

Line 5-11 is the classification step. Process $i$ is classified as a master process if the size of the union of valid values obtained in the read step is greater than its label $k$, otherwise, it is classified as a slave process. If it is classified as a slave process, it returns its input value set. If it is classified as a master process, process $i$ performs a write-read step by sending a $master$ message which includes the set of values it uses to do classification to all and wait for $n-f$ valid $mack$ message back at line 7. Similar to line 3, a $mack$ message is valid if each value contained in it has been reliably delivered with correct label. When a process receives a $master$ message with value set $T$ and label $k$ at round $r$, it first waits until all values in $T$ are reliably delivered. Then it sends back a $mack$ message along with the set of values reliably delivered with label $k$ at round $r$. The waiting is used to ensure that each value in $T$ is valid, i.e., be reliably delivered, because a Byzantine process can send arbitrary values in its $master$ message at line 7. By a similar reasoning as line 3, process $i$ will eventually obtain valid $mack$ message from at least $n - f$ different processes. After the write-read step, at line 8, process $i$ updates its value set to be the union of values obtained at line 7. 

 \begin{figure}[htbp] 
\fbox{\begin{minipage}[t]  {5.6in}
\noindent
{\bf function } $Valid(j, type, pf, v, k, r)$ for process $i$:\\ 
\h {\bf if} $(type = ``write" \wedge \lnot isSlave(j, k, r) \wedge v \subseteq S_i[k])$ \\
\h\h $\vee~ (type = ``write" \wedge isSlave(j, k,r) \wedge \text{BRB\_deliver}(j, ``write", -, v, LB_i^{r - 1}[j], r - 1)$ \\
\h\h\h\h $~\wedge ~pf[i] = RT_i^{r - 1}[j] \wedge |\bigcup \limits_{j = 1}^{n} pf[j]| \leq LB_i^{r - 1}[j])$ \\ \\
\h\h $\vee ~ (type = ``read" \wedge \text{BRB\_deliver}(j, ``write", -, -, k, r))$\\
\h\h {\bf return} $True$\\
\h {\bf else}\\
\h \h {\bf return} $False$ \\
\h {\bf endif} \\

 {\bf function} $isSlave(j, k, r)$ for process $i$: \\
\h {\bf if} $k = LB_i^{r - 1}[j] - \frac{f}{2^{r}}$ \\
\h\h {\bf return} $True$ \\
\h {\bf else} \\
\h\h {\bf return} $False$
\end{minipage}
} 
\caption{\textit{The Valid Function} \label{fig:valid_function}}
\end{figure}

The \textit{valid} function is defined in Fig. \ref{fig:valid_function}. In the this function, we first consider the \quotes{write}
messages. If the message has been sent by a process that claims to be a master, then it is considered valid if the value $v$ in this message is contained in the safe value set $S_i[k]$. 
If the message has been sent by a process that claims to be a slave, 
then process $i$ checks (1) whether process $i$ has BRB\_delivered the \quotes{write} message containing the same value at the previous round,
(2) whether the $i^{th}$ entry in $pf$ array matches the value process $j$ read from $i$ in the previous round, and (3) whether the the number 
of values contained in the proof $pf$ is at most $k$. 
The condition (1) ensures that a slave process sends the same value as the previous round since a correct slave process must keep its value same as in the previous round. The condition (2) ensures that the proof
sent by the slave process uses values that it read at round $r-1$.
The condition (3) checks that the sender classified itself correctly.

If the message is a \quotes{read} with label $k$ at round $r$, process $i$ considers it as valid if it BRB\_deliverd a \quotes{write} message with label $k$ at round $r$ from the sender. This is used to make sure that the sender (possibly Byzantine) must complete its write step in line 1-2 before trying to read at line 3-4. 

The $isSlave$ function invoked in the \textit{valid} function simply checks whether the label of the sender matches the label update rule by comparing it with the label at previous round.


\subsection{Proof of Correctness}
We first define the notion of \textit{committing} a message. 
\begin{definition}
We say a process {\bf commits} a message if it reliably broadcasts the message and the message is reliably delivered. We say a process {\bf commits} a message at time $t$ if this message is reliably delivered by the first process at time $t$.
\end{definition}

By properties of reliable broadcast, we observe that each process (possibly Byzantine) can commit at most one \quotes{write} message and at most one \quotes{read} message at each round.

\begin{table*}[h]
\caption{Notations}
\begin{tabular}{ |c|c|} 
 \hline
 \textbf{Variable} & \textbf{Definition} \\ \hline
 $G$ & A group of processes at round $r$ with label $k$ \\ \hline 
 $slave(G)$ & The slave subgroup of $G$, i.e., the processes with label $s(k,r)$ at round $r + 1$ \\ \hline 
 $master(G)$ & The master subgroup of $G$, i.e., the processes with label $m(k,r)$ at round $r + 1$ \\ \hline 
 $V_i^r$& The value set of process $i$ at the beginning of round $r$ \\ \hline
 $S_i^r$ & \makecell{The safe value map of process $i$ at the beginning of round $r$ \\ $S_i^r[k]$ is the safe value set of process $i$ for group $k$ at the beginning of round $r$}\\ \hline 
 $U_{k}^r$ & The set of admissible values for group $k$ at round $r$ \\ \hline 
\end{tabular}
\label{tab:variables}
\end{table*}

Define $s(k,r) = k - \frac{f}{2^{r + 1}}$ and $m(k, r) = k + \frac{f}{2^{r + 1}}$. The variables we use in the proof are shown in Table. \ref{tab:variables}. Consider the classification step in group $k$ at round $r$. The following lemma shows that if a Byzantine process wants to commit a \quotes{write} message $m$ at round $r + 1$ with a slave label, then it must commit a \quotes{write} message $m'$ which contains the same value as $m$ and a \quotes{read} message at round $r$ with label $k$. Also, it must commit its \quotes{read} message before its \quotes{write} message at round $r$ with label $k$. 
\begin{lemma} \label{lem:must_read_write} 
Suppose that process $i$ (possibly Byzantine) commits a write message \\
$(i, ``write", -, V_i,s(k,r), r + 1)$. Then \\
1) The message $(i, ``read", -, -, k, r)$ and the message $(i, ``write", -, V_i, k, r)$ must be committed by process $i$.\\
2) Let $t$ denote the time that message $(i, ``read", -,-,k,r)$ is committed. Then, the message $(i, ``write", -, V_i, k, r)$ must have been reliably delivered by at least $\lfloor \frac{n + f}{2} \rfloor + 1 - f$ correct processes before time $t$.  
\begin{proof}
For correct process $i$, the claim is obvious. 

Suppose $i$ is Byzantine. 
For it to commit a message with label $s(k,r)$ at round $r + 1$, its broadcast message has to be echoed by at least $\lfloor \frac{n + f}{2} \rfloor + 1$ different processes. Thus, it has to prove the values it read at round $r$ to at least $\lfloor \frac{n + f}{2} \rfloor + 1 - f$ different correct processes, which implies that it must commit $(i, ``read", -,-,k,r)$.  
For its read message to be reliably delivered, by the condition to echo a read message, we know that process $i$'s write message with label $k$ and value $V_i$ at round $r$ must be reliably delivered by at least $\lfloor \frac{n + f}{2} \rfloor + 1 - f$ correct processes. 

\end{proof}
\end{lemma}

The above lemma guarantees that if a Byzantine process wants to introduce some values in the admissible values of a slave group of group $k$, then it must first complete its write step and then complete its read step. Enforcing this order guarantees that the last slave process (possibly Byzantine) which completes its write step at line 2 must be able to read all values committed by slave processes in its read step at line 3. This set of values are exactly the set of admissible values for the slave group, since a slave process can only commit a \quotes{write} message which contains the same value as its \quotes{write} message at previous round by the \textit{valid} function. Then, since this last slave process is a valid slave process, which is verified in the {\em valid} function, we have that the union of all admissible values for the slave group has size at most $k$. 

The following lemma shows that the classifier guarantees the properties we defined. 
\begin{restatable}{lemma}{clsLemma}
\label{lem:cls} 
Let $G$ be a group at round $r$ with label $k$. Let $L$ and $R$ be two nonnegative integers such that $L < k \leq R$. If $L < |V_i^{r}| \leq R$ for each correct process $i \in G$, and $|U_{k}^r| \leq R$, then \\
(p1) For each correct process $i \in master(G)$, $k < |V_i^{r + 1}| \leq R$\\
(p2) For each correct process $i \in slave(G)$, $L < |V_i^{r + 1}| \leq k$\\
(p3) $U_{s(k,r)}^{r + 1} \subseteq U_{k}^r$ \\
(p4) $U_{m(k,r)}^{r + 1} \subseteq U_{k}^r$ \\
(p5) $|U_{m(k,r)}^{r + 1}| \leq R$ \\
(p6) $|U_{s(k,r)}^{r + 1}| \leq k$\\
(p7) For each correct process $j \in master(G)$, $U_{s(k,r)}^{r + 1} \subseteq V_j^{r + 1}$   \\
(p8) Each correct process $i \in slave(G)$ can commits its value set at round $r + 1$, i.e., $V_i^{r + 1} \subseteq U_{s(k,r)}^{r + 1}$ \\
(p9) Each correct process $j \in master(G)$ can commit its value set at round $r + 1$, i.e., $V_j^{r + 1} \subseteq U_{m(k,r)}^{r + 1}$ \\
(p10) $|\cup \{V_i^{r + 1} ~| ~ i \in slave(G) \cap C\}| \leq k$ \\
(p11) $|\cup \{V_i^{r + 1} ~| ~ i \in master(G) \cap C\}| \leq R$ 
\end{restatable} 
\begin{proof}
{\bf (p1)-(p2)} : Immediate from the classifier procedure. 

{\bf (p3):} A slave process can only commit the write message that it has reliably broadcast at previous round. Thus, $U_{s(k,r)}^{r + 1} \subseteq U_{k}^r$. 

{\bf (p4):} The safe value set of each correct process for group $m(k, r)$ is the union of values reliably broadcast by processes in group $k$ at round $r$. Thus, $U_{m(k,r)}^{r + 1} \subseteq U_k^{r}$.

{\bf (p5):} Immediate from $(p4)$ 

{\bf (p6)}: Consider group $s(k,r)$ at round $r + 1$. From Lemma \ref{lem:unique_label}, we know that this group must be the slave group of group $k$ at round $r$. Let $P$ denote the set of processes who commit a write message at round $r + 1$ with label $s(k,r)$. For each process $i \in P$, let $(i, ``write", pf_i, V_i, s(k, r), r + 1)$ denote the message that is committed by process $i$. Then, $U_{s(k,r)}^{r + 1} = \cup \{V_i ~ |~ i \in P\}$. 
From  part 2) of Lemma \ref{lem:must_read_write}, we have that for each process $i \in P$ the message $(i, ``write", -,V_i, k, r)$ must have been reliably delivered by at least $\lfloor \frac{n + f}{2} \rfloor + 1 - f$ correct processes. 
Let process $l \in P$ be the last process such that its write message $(l, ``write", -, V_l, k, r)$ is reliably delivered by at least $\lfloor \frac{n + f}{2} \rfloor + 1 - f$ different correct processes. Let $Q_l$ denote the set of correct processes which echoed message $(l, ``write", pf_l, V_l, s(k,r), r + 1)$. We have $|Q_l| \geq \lfloor \frac{n + f}{2} \rfloor + 1 - f$. By the condition of echoing a \quotes{write} message, we have $pf_l[q] = RT_q^r[l]$ for each $q \in Q_l$. 

 Consider an arbitrary process $p \in P$. Let $Q_p$ denote the set of the first $\lfloor \frac{n + f}{2} \rfloor + 1 - f$ correct processes which reliably delivered message $(p, ``write", -, V_p, k, r)$. Since $2 (\lfloor \frac{n + f}{2} \rfloor + 1 - f) > n - f$, there exists a correct process $s \in Q_l \cap Q_p$. Let $t_0$ denote the time that process $s$ sets its $RT_s^r[l]$ as $ACV_s^{r}[k]$. This happens after $(l, ``read", -,-,k,r)$ is reliably delivered by process $t$. From part 2) Lemma of \ref{lem:must_read_write}, the message $(l, ``write", -, V_l, k, r)$ must be reliably delivered by at least $\lfloor \frac{n + f}{2} \rfloor + 1 - f$ correct processes before time $t_0$. Since $l$ is the last process in $P$ such that its write message  $(l, ``write", -, V_l, k, r)$ is reliably delivered by at least $\lfloor \frac{n + f}{2} \rfloor + 1 - f$ different correct processes, process $p$'s write message $(p, ``write", -, V_p,k ,r)$ must be reliably delivered by each process in $Q_p$ before time $t_0$. Hence, when process $s$ sets $RT_s^r[l]$ as $ACV_s^r[k]$, the value set $V_p$ has been added into $ACV_s^r[k]$ by process $s$. Thus, $V_p \subseteq RT_s^r[l]$. Since $p$ is an arbitrary process in $P$, we have that for each process $p \in P$, there exists a process $s \in Q_l$ such that $V_p \subseteq RT_s^r[l]$. Hence, $\cup \{V_p ~|~ p \in P\} \subseteq \cup \{RT_q^r[l] ~|~ q \in Q\} \subseteq \bigcup \limits_{j = 1}^{n} pf_l[j]$. Since $|\bigcup \limits_{j = 1}^{n} pf_l[j]| \leq k$, we have $|U_{s(k,r)}^{r + 1}| = |\cup \{V_p ~|~ p \in P\}| \leq k$. 

{\bf (p7)}: Consider group $s(k,r)$ at round $r + 1$. Let $P$ denote the set of processes who commit a \quotes{write} message at round $r + 1$ with label $s(k,r)$. For each process $i \in P$, let $(i, ``write", pf_i, V_i, s(k, r), r + 1)$ denote the message that is commitd by process $i$. Then, $U_{s(k,r)}^{r + 1} = \cup \{V_i ~ |~ i \in P\}$. We show that for each $i \in P$, $V_i \subseteq V_j^{r + 1}$. From Lemma \ref{lem:must_read_write}, we have the message $(i, ``write", -,V_i,k, r)$ must have been reliably delivered by at least $\lfloor \frac{n + f}{2} \rfloor + 1 - f$ correct processes. Let $P$ denote the set of correct processes which echoed for message $(i, ``write", pf_i, V_i, s(k, r), r + 1)$. We have $|P| \geq \lfloor \frac{n + f}{2} \rfloor + 1 - f$. Then, process $i$'s \quotes{read} message $(i, ``read", -, -,k, r)$ and write message $(i, ``write", -, V_i, k, r)$ must be reliably delivered by each process in $P$. Also, by the echoing condition, we have $pf_i[p] = ACV_p^r[i]$ for each $p \in P$ and $|\bigcup \limits_{j = 1}^{n} pf_i[j]| \leq k$. Thus $V_i \subseteq ACV_p^{r}[k]$ for each $p \in P$.\\
Let $Q$ denote the set of correct processes which delivered the message $master(T_j^r, k, r)$ sent by correct process $j \in master(G)$ at line 7 of round $r$ and we have $|Q| \geq n - 2f$. Since $\lfloor \frac{n + f}{2} \rfloor + 1 - f + n - 2f > n - f$, there exists a correct process $s  \in P \cap Q$ such that it reliably delivers $(i, ``write", -, V_i, k, r)$ of process $i$ and message $master(T_j^r, k, r)$ from process $j$. We show that process $s$ must deliver $(i, ``write", -, V_i, k, r)$ before $master(T_j^r, k, r)$. Suppose that process $t$ delivers $master(T_j^r, k, r)$ before $(i, ``write", -, V_i, k, r)$ for contradiction. Then, we have $T_j^r \subseteq ACV_t^r[k]$ when process $t$ delivers $(i, ``write", -, V_i, k, r)$. We also have that process $t$ must reliable deliver $(i, ``read",-,-,k,r)$ after $(i, ``write", -, V_i, k, r)$. Thus, $T_j^r \subseteq ACV_t^r[k]$ when $s$ delivers $(i, ``read",-,-,k,r)$. Since $ACV_s^r[k] = pf_i[s]$, we have $T_j^r \subseteq pf_i[s]$. Since process $j$ is correct and $j \in master(G)$, then $|T_j^r| > k$. Thus, $|\bigcup \limits_{j = 1}^{n} pf_i[j]| > k$, contradiction. Therefore, process $t$ delivers $master(T_j^r, k, r)$ before $(i, ``write", -, V_i, k, r)$. Then, when $j$ receives $mack(R_t,r)$ from $t$, we must have $V_i \subseteq R_t$. Thus, $V_i \subseteq V_j^{r + 1}$. 

{\bf (p8):} Since process $i$ is correct, at round $r$, it must read from at least $n - 2f$ correct processes. Let $Q$ denote this set of correct processes. Then, at round $r + 1$, each process in $Q$ will echo the reliable broadcast message of process $i$. Thus, there will be at least $n - 2f$ echo messages. Since $f < \frac{n}{5}$, we have $n - 2f \geq \lfloor \frac{n + f}{2} \rfloor + 1$. Hence, eventually the message of $i$ will be reliable delivered. 

{\bf (p9):} Consider round $r$, from line 8, we know that $V_j^{r + 1} \subseteq ACV_j^r[k]$. From the property of reliable broadcast, any value $v \in V_j^{r + 1}$ is reliably broadcast by some process and will eventually be reliably delivered by each correct process. Hence, value $v$ will be included into the safe value set of each correct process for the group with label $m(k, r)$. Thus, at round $r + 1$, $V_j^{r + 1}$ will be eventually reliable delivered by each correct process.  

{\bf (p10):} Implied by $(p8)$ and $(p6)$. 

{\bf (p11):} Implied by $(p9)$ and $(p5)$.
\end{proof}

The following lemma shows that the value set of a correct process is non-decreasing. 
\begin{lemma} \label{lem:non_dec}
For any correct process $i$ and round $r$, $V_i^r \subseteq V_j^{r + 1}$. 
\begin{proof}
A slave process keeps its value set unchanged and a master process updates its value set to be the set of reliably delivered values which contains its own value set. 
\end{proof}
\end{lemma}

The following lemma is used later to show that processes in the same group at the end of the algorithm must have the same set of values. A similar lemma is given in \cite{zheng2018lattice, zheng2018linearizable}. We defer its detailed proof in Appendix \ref{app:bound_lemma}. 

\begin{restatable}{lemma}{boundLemma} \label{lem:bound}
Let $G$ be a group of processes at round $r$ with label $k$. Then \\
(1) for each correct process $i \in G$, $k - \frac{f}{2^r} \leq |V_i^{r}| \leq k + \frac{f}{2^r}$ \\
(2) $|U_{k}^{r}| \leq k + \frac{f}{2^r}$
\end{restatable}
\begin{proof}
By induction on round number $r$ and apply $(p1)$, $(p2)$, $(p5)$ and $p(6)$ of Lemma \ref{lem:cls}. 
\end{proof}

\begin{lemma}\label{lem:same_group}
Let $i$ and $j$ be two correct processes that are within the same group $G$ with label $k$ at the beginning of round $\log f + 1$. Then $V_i^{\log f + 1}$ and $V_j^{\log f + 1}$ are equal.
\begin{proof}
Let $G'$ be the parent of $G$ with label $k'$. Assume without loss of generality that $G = M(G')$. The proof for the case $G = S(G')$ follows in the same manner. Since $G'$ is a group at round $\log f$, by Lemma \ref{lem:bound}, we have: \\
 (1) for each correct process $p \in G'$, $k' - 1 < |V_p^{\log f}| \leq k' + 1$, and\\
 (2) $|U_{k'}^{\log f}| \leq k' + 1$

 Since $i \in G'$ and $j \in G'$, (1) and (2) hold for both process $i$ and $j$. By the assumption that $G = M(G')$, process $i$ and $j$ execute the \textit{Classifier} procedure with label $k'$ and are both classified as \textit{master}. Let $L = k' - 1$ and $R = k' + 1$, then by applying Lemma \ref{lem:cls}($p1$) we have $k' < |V_i^{\log f + 1}| \leq k' + 1$ and $k' < |V_j^{\log f + 1}| \leq k' + 1$, thus $|V_i^{\log f + 1}| = |V_j^{\log f + 1}| = k' + 1$. By ($p11$) of Lemma \ref{lem:cls}, we have $|\cup \{V_i^{\log f + 1}, V_j^{\log f + 1}\}| \leq k' + 1$. Thus, $V_i^{\log f + 1} = V_j^{\log f + 1}$. Therefore, $V_i^{r}$ and $V_j^{r}$ are equal at the beginning of round $\log f + 1$.
\end{proof}
\end{lemma}

\begin{lemma}\label{lem:comp_logf}
(Comparability) For any two correct process $i$ and $j$, $y_i$ and $y_j$ are comparable.
\begin{proof}
If process $i$ and $j$ are in the same group at the beginning of round $\log f + 1$, then by Lemma \ref{lem:same_group}, $y_i = y_j$. 
Otherwise, let $G$ be the last group that both $i$ and $j$ belong to. Suppose $G$ is a group with label $k$ at round $r$. Suppose $i \in slave(G)$ and $j \in master(G)$ without loss of generality. Then, $V_i^{\log f + 1} \subseteq U_{s(k,r)}^{r + 1} \subseteq V_j^{r + 1} \subseteq V_j^{\log f + 1}$, by ($p8$), ($p6$) ($p7$) and ($p5$) of Lemma \ref{lem:cls} and Lemma \ref{lem:non_dec}. 
\end{proof}
\end{lemma}

\logTheoremAsync*
\begin{proof}
{\bf Downward-Validity}. After the initial round, $x_i \in V_i^1$. By Lemma \ref{lem:non_dec}, we have $V_i^1 \subseteq V_i^{\log f + 1}$. Thus $x_i \in V_i^{\log f + 1}$. Since  $y_i = \sqcup \{v \in V_i^{\log f + 1}\}$, we have $x_i \leq y_i$.  

{\bf Comparability} follows from Lemma \ref{lem:comp_logf}. 

{\bf Upward-Validity}. In the initial round, reliable broadcast is used to construct the safe value set for the initial group with label $k_0$. By property of reliable broadcast, each Byzantine process can introduce at most one value into the safe value set for group $k_0$. After the initial round, all admissible values for each group must be subset of the values reliably delivered at the initial round. Thus, the union of all value sets held by correct processes must subset of the union of $\{x_i ~|~ i \in C\}$ and a set $B \subset X$ such that $|B| \leq f$. Therefore, $\sqcup \{y_i ~ | ~ i \in C\} \leq \sqcup(\{x_i ~ | ~ i \in C\} \cup B)$, where $B \subset X $ and $|B| \leq t$.
\end{proof}

\begin{remark}
The reliable broadcast primitive we use can be replaced by a more efficient one proposed by Imbs et al~\cite{imbs2016trading}, which can only tolerate $f < \frac{n}{5}$ Byzantine failures but takes 2 asynchronous communication rounds. This suffices for our application. 
\end{remark}

\section{$O(\log f)$ Rounds Algorithm for Authenticated BLA Problem} 
In this section, we present an $O(\log f)$ rounds algorithm for the BLA problem in authenticated (i.e., assuming digital signatures and public-key infrastructure) setting that can tolerate $f < \frac{n}{3}$ Byzantine failures by modifying the Byzantine tolerant classifier procedure in previous section. The Byzantine classifier procedure in authenticated setting is shown in Fig. \ref{fig:auth_classifier}. The primary difference lies in what a process does when it reliably delivers some message and the validity condition for echoing a broadcast message. The basic idea is to let a process sign the $ack$ message that it needs to send. Each process uses the set of signed $ack$ messages as proof of its completion of a write step or read step. In this section, we use $\langle x \rangle _i$ to denote a message $x$ signed by process $i$, i.e., $\langle x \rangle _i = \langle x, \sigma \rangle $, where $\sigma$ is the signature produced by process $i$ using its private signing key. We say a message is correctly signed by process $i$ if the signature within the message is a correct signature produced by process $i$. 

\subsection{The Authenticated Byzantine Tolerant Classifier}
The classifier in the authenticated setting is shown in Fig. \ref{fig:auth_classifier}. The primary difference between the classifier in previous section and the authenticated classifier is that in the authenticated classifier each process uses signed messages as proof of its group identity. 

\begin{figure}[htbp]
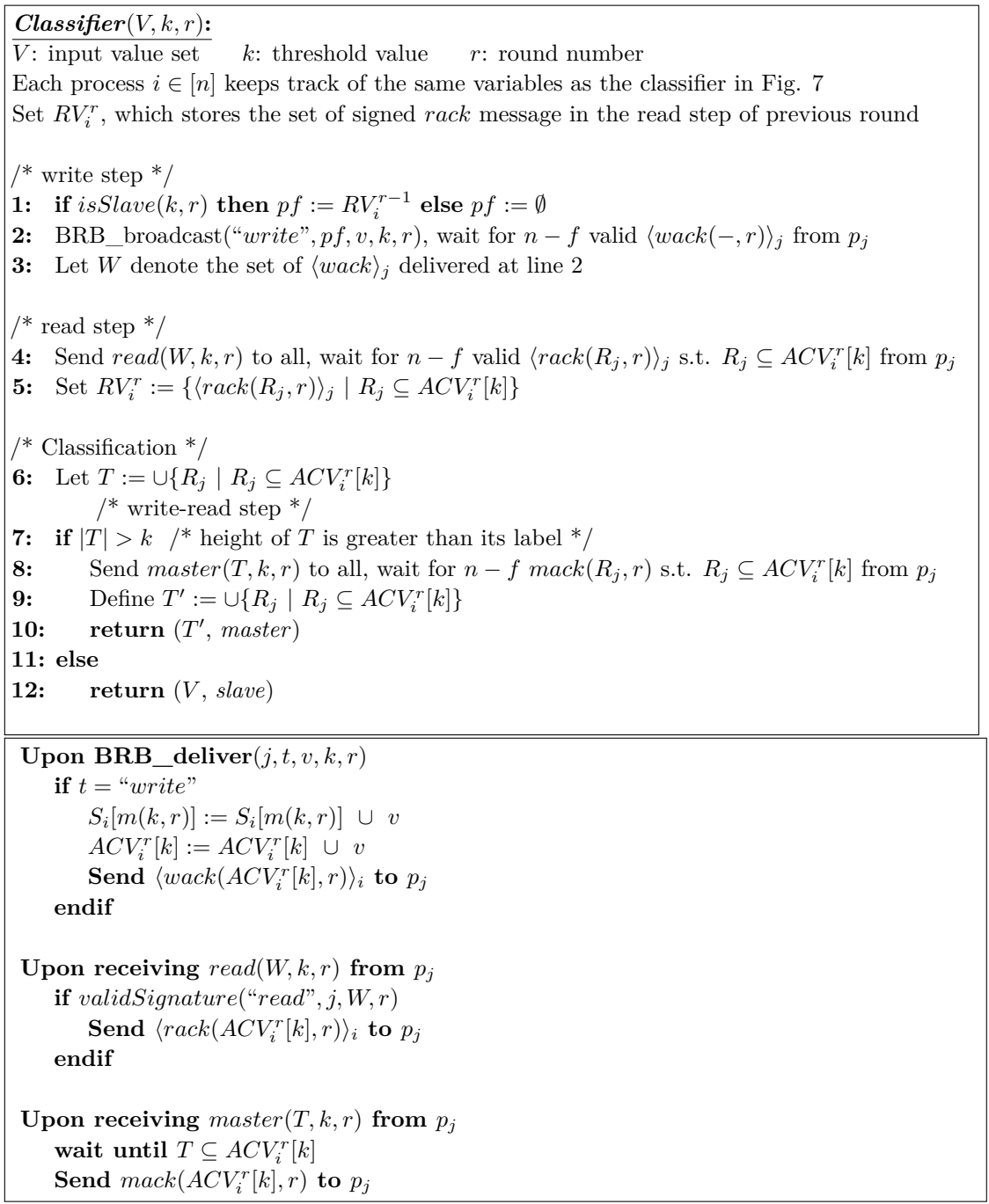
 
\fbox{\begin{minipage}[t]  {5.6in}
\noindent
\underline{{\bf \textit{Classifier}$(V, k, r)$:}} \\
$V$: input value set \h $k$: threshold value \h $r$: round number\\
Each process $i \in [n]$ keeps track of the same variables as the classifier in Fig. \ref{fig:classifier}\\
Set $RV_i^r$, which stores the set of signed $rack$ message in the read step of previous round \\  

/* write step */\\
{\bf 1: } {\bf if} $isSlave(k, r)$ {\bf then} $pf := RV_i^{r - 1}$ {\bf else} $pf := \emptyset$ \\
{\bf 2: } BRB\_broadcast$(``write", pf, v, k, r)$, wait for $n - f$ valid $\langle wack(-,r) \rangle _j$ from $p_j$\\
{\bf 3: } Let $W$ denote the set of $\langle wack \rangle _j$ delivered at line 2\\

/* read step */\\
{\bf 4: } Send $read(W, k, r)$ to all, wait for $n - f$ valid $\langle rack(R_j, r) \rangle _j$ s.t. $R_j \subseteq ACV_i^r[k]$ from $p_j$ \\
{\bf 5: } Set $RV_i^r := \{\langle rack(R_j, r) \rangle _j ~| ~ R_j \subseteq ACV_i^r[k]\}$\\

/* Classification */\\
{\bf 6: } Let $T := \cup \{R_j ~| ~ R_j \subseteq ACV_i^r[k]\}$ \\
\h\h ~ /* write-read step */ \\
{\bf 7: } {\bf if} {$|T| > k$} ~/* height of $T$ is greater than its label */ \\
{\bf 8: } \h Send $master(T, k, r)$ to all, wait for $n- f$ $mack(R_j, r)$ s.t. $R_j \subseteq ACV_i^r[k]$ from $p_j$\\
{\bf 9: } \h Define $T' := \cup \{R_j ~| ~ R_j \subseteq ACV_i^r[k]\}$ \\
{\bf 10:} \h {\bf return} ($T'$, \textit{master})\\
{\bf 11:} {\bf else}\\
{\bf 12:} \h {\bf return} ($V$, \textit{slave})\\
\end{minipage}
}

\fbox{
\begin{minipage}[t] {5.6in}
\bf{Upon} BRB\_deliver$(j, t, v, k, r)$\\
\h {\bf if} $t =``write"$ \\
\h\h $S_i[m(k, r)] := S_i[m(k, r)] ~\cup ~v$ \\
\h\h $ACV_i^{r}[k] := ACV_i^{r}[k] ~\cup~ v$\\
\h\h Send $\langle wack(ACV_i^r[k], r) \rangle _i$ to $p_j$\\
\h {\bf endif} \\

{\bf Upon} receiving $read(W, k, r)$ from $p_j$ \\
\h {\bf if} $validSignature(``read", j, W, r)$ \\ 
\h\h Send $\langle rack(ACV_i^{r}[k], r) \rangle _i$ to $p_j$\\ 
\h {\bf endif} \\

{\bf Upon} receiving $master(T, k, r)$ from $p_j$ \\
\h {\bf wait until} $T \subseteq ACV_i^{r}[k]$ \\
\h Send $mack(ACV_i^{r}[k], r)$ to $p_j$
\end{minipage} 
} 
\caption{\textit{The Authenticated Byzantine Tolerant Classifier} \label{fig:auth_classifier}}
\end{figure}

At lines 1-2, each process writes its current value set by using the {\bf BRB\_broadcast} procedure to send a \quotes{write} message. If the process is a slave process, it also includes the set of at least $n - f$ signed $rack$ messages it received at the previous round as a proof that it is indeed classified as a slave. At line 2, each process waits for correctly signed $wack$ message from at least $n- f$ different processes. This set of signed $wack$ message is used as the proof of its completion of the write step when this process tries to read from other processes. When a process BRB\_delivers a \quotes{write} message, it performs similar steps as the algorithm in previous section except that it sends a signed $wack$ message back. 

At line 4-5, each process reads from at least $n - f$ processes. Different from the classifier procedure in previous section, each process directly sends a read message along with the set of correctly signed $wack$ messages obtained at line 2 to all (instead of using the {\bf BRB\_broadcast} procedure). When a process receives a \quotes{read} message with label $k$ for round $r$, if uses the \textit{validSignature} function to check whether the \quotes{read} message contains correctly signed $wack$ message for round $r$ from at least $n - f$ different processes. If so, it sends back to the sender a signed $rack$ message along with the reliably delivered values with label $k$ at round $r$. This ensures that if a process (possibly Byzantine) tries to read from correct processes, it must complete its write step first. 

The classification step from line 6-12 is the same as the classification step of the algorithm in previous section. A mater process performs a write-read step by sending a $master$ message along the set of value obtained at line 6. Then it waits for $n - f$ valid $mack$ messages and updates its value set to be the set of values contained in these messages. When a process receives a $master$ message, it performs the same steps as in the classifier in previous section.  

The \textit{valid} function is different from the one given in previous section. First, only \quotes{write} messages are reliably broadcast. Second, the proof is a set of signed $rack$ messages instead of an array in previous section. To verify the proof, the valid function invokes the \textit{validSignature} function to check whether the proof contains correctly signed $rack$ message for previous round from at least $n - f$ different processes.

\begin{figure}[htb] 
\fbox{\begin{minipage}[t]  {5.4in}
\noindent

{\bf function } $Valid(j, type, pf, v, k, r)$:\\ 
\h {\bf if} $(type = ``write" \wedge \lnot isSlave(j, k, r) \wedge v \subseteq S_i[k])$ \\
\h\h $\vee~ (type = ``write" \wedge isSlave(j, k,r) \wedge \text{BRB\_deliver}(j, ``write", -, v, LB_i^{r - 1}[j], r - 1)$ \\
\h\h\h\h $~\wedge validSignature(``write", pf, r) \wedge pf ~\text{contains at most $k$ distinct values}$\\
\h\h {\bf return} $True$\\
\h {\bf else}\\
\h \h {\bf return} $False$ \\
\h {\bf endif} \\

{\bf function} $validSignature(type, pf, r)$: \\
\h {\bf if} ($type = ``write" \wedge pf$ contains correctly signed $rack(-,r - 1)$ from $n - f$ processes) \\
\h \h $\vee ~$ ($type = ``read" \wedge pf$ contains correctly signed $wack(-,r)$ from $n - f$ processes) \\
\h\h {\bf return} $True$ \\
\h {\bf else} \\
\h\h {\bf return} $False$ \\
\h {\bf endif}
\end{minipage}
} 
\caption{\textit{The Valid Function} \label{fig:classifier}}
\end{figure}

For the proof of correctness, we just need to prove the classifier procedure satisfies the properties given Lemma \ref{lem:cls} under the assumption that $f < \frac{n}{3}$. The proof of $(p6)$ and $(p7)$ is similar to the proof in previous section. Thus, we do not give the formal mathematical proof. 
\begin{lemma}
Properties $(p1)-(p11)$ of Lemma \ref{lem:cls} hold for the classifier shown in Fig. \ref{fig:auth_classifier}.

\begin{proof}
{\bf (p1)-(p5)} : similar to the proofs given in Lemma \ref{lem:cls}.  

{\bf (p6)}: Consider group $s(k,r)$ at round $r + 1$. Let $Q$ denote the set of processes who commit their \quotes{write} message with label $s(k,r)$ at round $r + 1$. For process $i$ (possibly Byzantine) to commit this message, this message has to be echoed by at least $\lfloor \frac{n + f}{2} \rfloor + 1$ different processes. By the condition of echo, process $i$ must have received correctly signed $rack(-,r)$ message from at least $n - f$ processes at line 4 of round $r$, among which there are at least $n - 2f$ correct processes. When a process tries to read at round $r$, it has to prove that it has received at least $n - f$ correctly signed $wack(-,r)$ messages. Thus, process $i$ must have committed a write message and this message has been reliably delivered by at least $n - 2f$ correct processes.  Let process $l \in Q$ be the last process in $Q$ such that its write message is reliably delivered by at least $n - 2f$ different correct processes. When process $l$ read from at least $n - 2f$ different correct processes at round $r$, each process in $Q$ must have written its value in at least $n - 2f$ different correct processes. Then, since $2 (n - 2f) > n - f$ by the assumption of $f < \frac{n}{3}$, process $l$ must have read all values reliably broadcast by processes in $Q$. These values are exactly the set $U_{s(k,r)}^{r + 1}$. Since $l$ is a slave process and the size of $RV_l^r$ must be at most $k$, we have $|U_{s(k,r)}^{r + 1}| \leq k$. 

{\bf (p7)}: Consider group $s(k,r)$ at round $r + 1$. For a process $i \in slave(G)$ to commit its value set at round $r + 1$, process $i$ must reliably broadcast the same value set as round $r$ and must be able to verify the value set $RV_i^r$ it reads at round $r$. By a similar argument as the proof for $(p6)$, process $i$'s value set must be reliably delivered by at least $n - 2f$ correct processes at round $r$ before its reading step. Let $P$ denote the set of correct processes which delivered process $i$'s value before its reading step at round $r$ and $|P| \geq n - 2f$. Let $Q$ denote the set of correct processes which delivered the $master$ message sent by process $j \in master(G)$ at line 11 of round $r$ and $|Q| \geq n - 2f$. Then, since $2 (n - 2f) > n - f$ by the assumption of $f < \frac{n}{3}$, there exists a correct process $k  \in P \cap Q$ such that it delivered both $V_i^r$ and $T_j^r$ at round $r$. Since process $i$ is a slave process, process $k$ must send $rack$ to process $i$ before sending $mack$ to process $j$, otherwise, process $i$ would be classified as master. Then, process $j$ must receive $V_j^r$ from process $k$ and include into $V_j^{r + 1}$. Thus, $U_{s(k,r)}^{r + 1} \subseteq V_j^{r + 1}$.  

{\bf (p8):} Since process $i$ is correct, at round $r$, it must have received at least $n - f$ correctly signed $rack$ message. Then, at round $r + 1$, it can prove to at least $\lfloor \frac{n + f}{2} \rfloor + 1$ correct processes. Hence, the write message of $i$ at round $r + 1$ will be reliably delivered by each correct process. 

{\bf (p9):} Consider round $r$, from line 8, we know that $V_j^{r + 1} \subseteq ACV_j^r[k]$. From the property of reliable broadcast, any value $v \in V_j^{r + 1}$ is reliably broadcast by some process with label $k$ at round $r$ and will eventually be reliably delivered by each correct process. Hence, value $v$ will be included into the safe value set of each correct process for the group $m(k, r)$. Then, eventually $V_j^{r + 1} \subseteq S_i[m(k,r)]$ for each correct process $i$. Thus, when process $j \in master(G)$ tries to reliable broadcast $V_j^{r + 1}$ at round $r + 1$, the echo condition eventually holds for each correct process $i$.
\end{proof}
\end{lemma}
\logTheoremAsyncAuth*

\section{Conclusion}

In this paper, we present an $O(\log f)$ rounds algorithm for the Byzantine lattice agreement problem in asynchronous systems which can tolerates $f < \frac{n}{5}$ Byzantine failures. We also give an $O(\log f)$ rounds algorithm for the authenticated setting that can tolerate $f < \frac{n}{3}$ Byzantine failures. One open problem left is to design an algorithm which has resilience of $f < \frac{n}{3}$ and takes $O(\log f)$ rounds.

\bibliography{ref}

\appendix 

\section{Proof of Lemma \ref{lem:bound}} \label{app:bound_lemma}
\boundLemma* 
\begin{proof}
By induction on $r$. Consider the base case with $r = 1$, $k = k_0 = n - \frac{f}{2}$. After the initial round, each process must receive at least $n - f$ different values and at most $n$ values by the property of reliable broadcast. Thus, $k_0 - \frac{f}{2} = n - f \leq |V_i^1| \leq n$. Any value in $U_k^r = U_{k_0}^1$ must be reliably delivered by at least one correct process at the initial round. Thus, $|U_{k_0}^1| \leq k_0 + \frac{f}{2} = n$. 

For the induction step, assume the above lemma holds for all groups at round $r - 1$. Consider an arbitrary group $G$ at round $r > 1$ with label $k$. Let $G'$ be the parent group of $G$ at round $r - 1$ with label $k'$. Consider the \textit{Classifier} procedure executed by all processes in $G'$ with label $k'$. By induction hypothesis, we have: 

(1) for each correct process $i \in G'$, $k' - \frac{f}{2^{r - 1}} < |V_i^{r - 1}| \leq k' + \frac{f}{2^{r - 1}}$ 

(2) $|U_{k'}^{r - 1}| \leq k' + \frac{f}{2^{r - 1}}$.

 Let $L = k' - \frac{f}{2^{r - 1}}$ and $R = k' + \frac{f}{2^{r - 1}}$, then (1) and (2) are exactly the conditions of Lemma \ref{lem:cls}. Consider the following two cases: 

 Case 1: $G = M(G')$. Then $k = k' + \frac{f}{2^r}$. From ($p1$) and ($p5$) of Lemma \ref{lem:cls}, we have:
 
 (1) for each correct process $i \in G$, $k - \frac{f}{2^r} < |V_i^r| \leq k + \frac{f}{2^r}$ 
 
 (2) $|U_{k}^r| \leq k + \frac{f}{2^r}$. 

 Case 2: $G = S(G')$. Then $k = k' - \frac{f}{2^r}$. Similarly, from ($p2$) and ($p6$) of Lemma \ref{lem:cls}, we have the same equations. 
 \end{proof}
 
\end{document}